\documentclass[11pt]{article}
\usepackage{amssymb}
\usepackage[perpage,symbol]{footmisc}
\usepackage{listings}
\usepackage{amsfonts}
\usepackage{enumerate}
\usepackage{amsmath}
\usepackage{cite}
\usepackage{array}
\usepackage{booktabs}

\begin{document}

\newtheorem{theorem}{Theorem}[section]
\newtheorem{corollary}[theorem]{Corollary}
\newtheorem{definition}[theorem]{Definition}
\newtheorem{proposition}[theorem]{Proposition}
\newtheorem{conjecture}[theorem]{Conjecture}
\newtheorem{lemma}[theorem]{Lemma}
\newtheorem{example}[theorem]{Example}
\newenvironment{proof}{\noindent {\bf Proof.}}{\rule{3mm}{3mm}\par\medskip}
\newcommand{\remark}{\medskip\par\noindent {\bf Remark.~~}}
\title{Kasami type codes of higher relative dimension}
\author{Chunlei Liu\footnote{Dept. of math., Shanghai Jiao Tong Univ., Sahnghai 200240, China, clliu@sjtu.edu.cn.}\ \footnote{Dengbi Technologies Cooperation Limited, Yichun 336099, China, 714232747@qq.com.}}
\date{}
\maketitle
\thispagestyle{empty}

\abstract{Let $m,n,d,e$ be fixed positive integers such that \[m=2n,~e=(n,d)=(m,d), ~1\leq k \leq \frac{n}{e}.\]
Let $s$ be a fixed maximum-length binary sequence of length $2^{m}-1$. For $0\leq j<k$, let
$s_j$ be the circular decimation of $s$ with decimation factor $2^{(\frac{n}{e}-j)d}+1$.
Then $s,s_1,\cdots,s_{k-1}$ are maximum-length binary sequences of length $2^{m}-1$, while $s_0$ is a maximum-length binary sequence of length $2^{n}-1$. Let $C$ be the ${\mathbb F}_2$-vector space generated by all circular shifts of $s,s_0,s_1,\cdots,s_{k-1}$. Then $C$ has an ${\mathbb F}_{2^n}$-vector space structure, and is of dimension $2k+1$ over ${\mathbb F}_{2^n}$. We regard $C$ as a  Kasami type code of relative dimension $2k+1$. The DC component distribution of $C$ is explicitly calculated out in the present paper.}

\noindent {\bf Key phrases}: Kasami code, cyclic code, alternating form


\section{\small{INTRODUCTION}}
\paragraph{}
Let $q$ be a prime power, and $C$ an $[n,k]$-linear code over ${\mathbb F}_q$. The weight of a codeword $c=(c_0,c_1,\cdots,c_{n-1})$ of $C$ is defined to be
\[{\rm wt}(c)=\#\{0\leq i\leq n-1|~c_i\neq0\}.\]
For each $i=0,1,\cdots,n$, define
\[A_i=\#\{c\in C\mid~{\rm wt}(c)=i\}.\]The sequence $(A_0,A_1,\cdots,A_n)$ is called the weight distribution of $C$.
Given a linear code $C$, it is challenging to determine its weight distribution.
The
weight distribution of Gold codes was determined by Gold \cite{Gold66, Gold67, Gold68}.
The
weight distribution of Kasami codes was determined by Kasami \cite{Kasami66}.
The
weight enumerators of Gold type and Kasami type codes of higher relative dimension were determined by Berlekamp  \cite{Ber} and Kasami \cite{Kasami71}. The
weight distribution of some new Gold type codes of higher relative dimension was determined by Liu \cite{Liu}.
The
weight distribution of the $p$-ary analogue of Gold codes was determined by Trachtenberg \cite{Tr}.
The
weight distribution of the circular decimation of the $p$-ary analogue of Gold codes with decimation factor $2$ was determined by Feng-Luo \cite{FL}.
The
weight distribution of the $p$-ary analogue of Gold type codes of relative dimension $3$ was determined by Zhou-Ding-Luo-Zhang \cite{ZDLZ}.
The
weight distribution of the circular decimation with decimation factor $2$ of the $p$-ary analogue of Gold type codes  of relative dimension $3$ was determined by Zheng-Wang-Hu-Zeng \cite{ZWHZ}.
The
weight distribution of the $p$-ary analogue of Kasami type codes of maximum dimension was determined by Li-Hu-Feng-Ge \cite{LHFG}. The
weight distribution of the $p$-ary analogue of Gold type codes of higher relative dimension was determined by Schmidt \cite{Sch}. The weight distribution of some other classes of cyclic codes was determined in the papers \cite{AL}, \cite{BEW}, \cite{BMC}, \cite{BMC10}, \cite{BMY}, \cite{De}, \cite{DLMZ}, \cite{DY}, \cite{FE}, \cite{FM}, \cite{KL}, \cite{LF}, \cite{LHFG}, \cite{LN}, \cite{LYL}, \cite{LTW}, \cite{MCE}, \cite{MCG}, \cite{MO}, \cite{MR}, \cite{MY}, \cite{MZLF}, \cite{RP}, \cite{SC}, \cite{VE}, \cite{WTQYX}, \cite{XI}, \cite{XI12}, \cite{YCD}, \cite{YXDL} and \cite{ZHJYC}.
\paragraph{}
Let $m,n,d,e$ be fixed positive integers such that \[m=2n,~e=(n,d)=(m,d), ~1\leq k \leq \frac{n}{e}.\]
Let $s$ be a fixed maximum-length binary sequence of length $2^{m}-1$. For $0\leq j<k$, let
$s_j$ be the circular decimation of $s$ with decimation factor $2^{(\frac{n}{e}-j)d}+1$.
Then $s,s_1,\cdots,s_{k-1}$ are maximum-length binary sequences of length $2^{m}-1$, while $s_0$ is a maximum-length binary sequence of length $2^{n}-1$. Let $C$ be the ${\mathbb F}_2$-vector space generated by all circular shifts of $s,s_0,s_1,\cdots,s_{k-1}$.
For each $\vec{a}\in\mathbb{F}_{2^{n}}\times{\mathbb F}_{2^m}^k$, define a quadratic form on the ${\mathbb F}_{2^e}$-vector space ${\mathbb F}_{2^m}$ by the formula
\[
Q_{\vec {a}}(x)={\rm  Tr}_{\mathbb{F}_{2^{n}}/\mathbb{F}_{2^e}}(a_0x^{2^{\frac{nd}{e}}+1})+\sum_{j=1}^{k-1}{\rm  Tr}_{\mathbb{F}_{2^{m}}/\mathbb{F}_{2^e}}(a_jx^{2^{(\frac{n}{e}-j)d}+1})+{\rm  Tr}_{\mathbb{F}_{2^{m}}/\mathbb{F}_{2^e}}(a_kx).
\]
Then \[C=\{c_{\vec a}=(c_{{\vec a},0},\cdots,c_{{\vec a},2^m-2})\mid~\vec{a}\in\mathbb{F}_{2^{n}}\times{\mathbb F}_{2^m}^k\},\]
where \[c_{{\vec a},i}={\rm Tr}_{{\mathbb F}_{2^e}/{\mathbb F}_2}(Q_{\vec a} (\pi^{-i})),\]
with $\pi$ being a primitive element of ${\mathbb F}_{2^m}$.
The correspondence $\vec{a}\mapsto c_{\vec a}$ defines an ${\mathbb F}_{2^n}$-vector space structure on $C$, and $C$ is of dimension $2k+1$ over ${\mathbb F}_{2^n}$. When $k=1$, $C$ is the Kasami code. So we call $C$ a Kasami type code of relative dimension $2k+1$. If $d=e=1$, $C$ is the code studied by Kasami \cite{Kasami71}.
\paragraph{}
One can prove the following.
\begin{theorem}\label{dcbound}
If $c\in C$ is nonzero, then \[{\rm DC}(c)\in\{-1,-1+\pm2^{\frac{m}{2}+je}\mid~j=0,1,2,\cdots,k-1\},\]
where
\[{\rm DC}(c)=\sum_{i=0}^{2^m-2}(-1)^{c_i}\]
is the DC component of $c=(c_0,c_1,\cdots,c_{2^m-2})\in C$.
\end{theorem}
The present paper is concerned with the frequencies
\begin{equation}\label{dcfrequencydef}\alpha_{r,\varepsilon}=\#\{0\neq c\in C\mid~{\rm DC}(c)=-1+\varepsilon 2^{m-\frac{er}{2}}\},~r=0,2,4,\cdots,\frac{m}{e}.\end{equation}
The main result of the present paper is the following.
\begin{theorem}\label{main}
For each $j=0,1,\cdots,k-1$, and for each $\varepsilon=\pm1$, we have
\[\alpha_{\frac{m}{e}-2j,\varepsilon}=\frac12(2^{m-2ej}+\varepsilon2^{\frac{m}{2}-ej})
\sum_{v=j}^{k-1}(-1)^{v-j}4^{e\binom{v-j}{2}}\binom{v}{j}_{4^e}
\binom{\frac{m}{2e}}{v}_{4^e}(2^{n(2k-1-2v)+ev}-1),\]
where $\binom{j}{i}_{q}$ is the Gaussian binomial coefficient.
\end{theorem}
From the above theorem one can deduce the following.
\begin{theorem}\label{balanced}We have
\[\begin{split}&\#\{c\in C\mid
~{\rm DC}(c)=-1\}\\&=2^{n(2k+1)}-1-
\sum_{v=0}^{k-1}(-1)^v(2^{n(2k-1-2v)+ev}-1)2^{m-ev(v+1)}\prod_{j=0}^{v-1}(2^m-4^{ej})\\
&\approx2^{n(2k+1)}\sum_{v=1}^{k-1}(-1)^{v-1}2^{-ev^2}.\end{split}\]
\end{theorem}
If $d=e=1$, then the weight enumerator of $C$ is determined by Kasami \cite{Kasami71}. However, some extra calculations are needed to explicitly write out the coefficients of the weight enumerator in \cite{Kasami71}.
\section{\small{ENTERING BILINEAR FORMS}}
 \paragraph{}
In this section we shall prove Theorem \ref{dcbound}.
Note that
\begin{equation}\label{dcexpsumrelation}1+{\rm DC}(c_{\vec a})=\sum_{x\in{\mathbb F}_{2^m}}(-1)^{{\rm Tr}_{{\mathbb F}_{2^e}/{\mathbb F}_2}(Q_{\vec a} (x))}.\end{equation}
It is well-known that
\begin{equation}\label{quadraticexpsum}
\sum_{x\in{\mathbb F}_{2^m}}(-1)^{{\rm Tr}_{{\mathbb F}_{2^e}/{\mathbb F}_2}(Q_{\vec a} (x))}=\begin{cases}
           0,&  2 \nmid {\rm rk}(Q_{\vec a}),\\
             \pm2^{m-e\cdot\frac{{\rm rk}(Q_{\vec a})}{2})},&2 |{\rm rk}(Q_{\vec a}).
           \end{cases}
\end{equation}
 \paragraph{}
Let \[
B_{\vec {a}}(x,y)=Q_{\vec a}(x+y)-Q_{\vec a}(x)-Q_{\vec a}(y).\]
Then
 \[
B_{\vec {a}}(x,y)={\rm  Tr}_{\mathbb{F}_{2^{n}}/\mathbb{F}_{2^e}}(a_0(xy^{2^{\frac{nd}{e}}}+x^{2^{\frac{nd}{e}}}y))+
\sum_{j=1}^{k-1}{\rm  Tr}_{\mathbb{F}_{2^{m}}/\mathbb{F}_{2^e}}(a_{j}(xy^{2^{(\frac{n}{e}-j)d}}+x^{2^{(\frac{n}{e}-j)d}}y)).\]
It is well-known that
\begin{equation}\label{rkbilinear}{\rm rk}(B_{\vec a})=\left\{
                             \begin{array}{ll}
                              {\rm rk}(Q_{\vec a}), & \hbox{} 2\mid{\rm rk}(Q_{\vec a}),\\
                               {\rm rk}(Q_{\vec a})-1, & \hbox{}2\nmid{\rm rk}(Q_{\vec a}).
                             \end{array}
                           \right.\end{equation}
 \paragraph{}
We now prove Theorem \ref{dcbound}.
By (\ref{dcexpsumrelation}), (\ref{quadraticexpsum}) and (\ref{rkbilinear}), it suffices to prove the following.
\begin{theorem}\label{rankbound}If $(a_0,a_1,\cdots,a_{k-1})\neq0$, then
\[{\rm rk}(B_{\vec a})\geq \frac{m}{e}-2(k-1).\]
\end{theorem}
\begin{proof} Suppose that $(a_0,a_1,\cdots,a_{k-1})\neq 0$. It suffices to show that
\[
 {\rm dim}_{\mathbb{F}_{2^e}}{\rm Rad}({B_{\vec {a}}} )\leq 2(k-1),
 \]
where
\[{\rm Rad}({B_{\vec {a}}})=\{x \in  \mathbb{F}_{2^{m}}\mid~B_{\vec a}(x,y)=0,~\forall y\in{\mathbb F}_{2^m}\}.\]
We have
\[\begin{split}
&{\rm Rad}({B_{\vec {a}}})=\{x \in  \mathbb{F}_{2^{m}}\mid~a_0x^{2^{\frac{nd}{e}}}+\sum_{j=1}^{k-1}(a_{j}^{2^{(j-\frac{n}{e})d}}x^{2^{(j-\frac{n}{e})d}}+a_{j}x^{2^{(\frac{n}{e}-j)d}})=0\}\\
&=\{x \in  \mathbb{F}_{2^{m}}\mid~a_0^{2^{(\frac{n}{e}+k-1)d}}x^{2^{k-1}}
+\sum_{j=1}^{k-1}(a_{j}^{2^{(k-1+j)d}}x^{2^{(k-1+j)d}}+a_{j}^{2^{(\frac{n}{e}+k-1)d}}x^{2^{(k-1-j)d}})=0\}.
\end{split}
\]Note that
\[
 \{x \in  \mathbb{F}_{2^{md/e}}\mid~a_0^{2^{(\frac{n}{e}+k-1)d}}x^{2^{k-1}}
+\sum_{j=1}^{k-1}(a_{j}^{2^{(k-1+j)d}}x^{2^{(k-1+j)d}}+a_{j}^{2^{(\frac{n}{e}+k-1)d}}x^{2^{(k-1-j)d}})=0\}.
 \]
is a subspace of $\mathbb{F}_{2^{md/e}}$ over ${\mathbb F}_{2^d}$ of dimension $\leq 2(k-1)$.
As $(m, d)=e$,
a basis of $\mathbb{F}_{2^{m}}$ over ${\mathbb{F}_{2^{e}}}$ is also a basis of $\mathbb{F}_{2^{md/e}}$ over $\mathbb{F}_{2^{d}}$.
It follows that
\[
 {\rm dim}_{\mathbb{F}_{2^e}}{\rm Rad}({B_{\vec {a}}})\leq 2(k-1).
 \]
The theorem is proved.\end{proof}
\section{\small{AN INVERSION FORMULA}}
Let $q$ be a prime power. In this section we shall prove an inversion formula for the symmetric van der Monte matrix
$(q^{ij})_{0\leq i,j\leq u}$. We begin with
the following well-known formula.
\begin{theorem}[$q$-binomial M\"{o}bius inversion formula] Suppose that
$u>v$. Then
the vector
$(\binom{i}{v}_q)_{i=v}^u$ is orthogonal to the vector
$((-1)^{u-i}q^{\binom{u-i}{2}}\binom{u}{i}_{q})_{i=v}^u$,
and
the vector
$(\binom{u}{i}_q)_{i=v}^u$ is orthogonal to the vector
$((-1)^{i-v}q^{\binom{i-v}{2}}\binom{i}{v}_{q})_{i=v}^u$.
\end{theorem}
We now prove the following.
\begin{theorem}[Inversion of a symmetric van der Monte matrix]We have
\[\sum_{j=0}^uq^{ij}x_j=y_i,~0\leq i\leq u\]
if and only if
\[x_j
=\sum_{v=j}^{u}(-1)^{v-j}q^{\binom{v-j}{2}}\binom{v}{j}_{q}\prod_{i=0}^{v-1}(q^{v}-q^i)^{-1}\sum_{i=0}^v(-1)^{v-i}q^{\binom{v-i}{2}}\binom{v}{i}_{q}y_i.\]
\end{theorem}\begin{proof}
Fix $0\leq v\leq u$. Consider the equation
\[\sum_{j=0}^ux_j\left(
                   \begin{array}{c}
                     1 \\
                     q^j \\
                     \vdots \\
                     q^{vj} \\
                   \end{array}
                 \right)=\left(
                           \begin{array}{c}
                             y_0 \\
                             y_1 \\
                             \vdots \\
                             y_v \\
                           \end{array}
                         \right)
\]
Multiplying on the left by the row vector $((-1)^{v-i}q^{\binom{v-i}{2}}\binom{v}{i}_{q})_{i=0}^v$, and applying the $q$-binomial formula, we arrive at
\[\sum_{j=v}^u x_j\prod_{i=0}^{v-1}(q^{j}-q^i)
=\sum_{i=0}^v(-1)^{v-i}q^{\binom{v-i}{2}}\binom{v}{i}_{q}y_i.\]
Dividing both sides by $\prod_{i=0}^{v-1}(q^{v}-q^i)$,
we arrive at
\[\sum_{j=v}^u x_j\binom{j}{v}_q
=\prod_{i=0}^{v-1}(q^{v}-q^i)^{-1}\sum_{i=0}^v(-1)^{v-i}q^{\binom{v-i}{2}}\binom{v}{i}_{q}y_i.\]
Applying $q$-binomial M\"{o}bius inversion formula, we arrive at
\[x_j
=\sum_{v=j}^{u}(-1)^{v-j}q^{\binom{v-j}{2}}\binom{v}{j}_{q}\prod_{i=0}^{v-1}(q^{v}-q^i)^{-1}\sum_{i=0}^v(-1)^{v-i}q^{\binom{v-i}{2}}\binom{v}{i}_{q}y_i.\]
The theorem is proved.\end{proof}
\section{\small{A PRODUCT FORMULA}}
Let $q$ be a prime power.
We begin with the following product formula.
\begin{theorem}\label{prodform}If $i\geq1$, then \[\sum_{j=0}^iq^{j}\binom{i}{j}_{q^2}
=\prod_{j=1}^{i}(1+q^{j}).
\]\end{theorem}
\begin{proof}
The product formula in the theorem is trivial if $i=1$. We now assume that $i\geq2$. By the $q$-binomial recursion formula,
\[\begin{split}&\sum_{j=0}^iq^{j}\binom{i}{j}_{q^2}\\
=&1+\sum_{j=1}^iq^{j}(\binom{i-1}{j}_{q^2}+\binom{i-1}{j-1}_{q^2}q^{2(i-j)})\\
=&\sum_{j=0}^{i-1}q^{j}\binom{i-1}{j}_{q^2}+q^{i}\sum_{j=0}^{i-1}\binom{i-1}{j}_{q^2}q^{i-1-j}\\
=&(1+q^i)\sum_{j=0}^{i-1}q^{j}\binom{i-1}{j}_{q^2}.\end{split}
\]
The theorem now follows by induction.
\end{proof}
We now prove the following product formula.
\begin{theorem}\label{twovsone}If $i\geq1$, then \[\binom{u}{i}_{q^2}\sum_{j=0}^iq^{j}\binom{i}{j}_{q^2}
=\binom{u}{i}_{q}\prod_{j=0}^{i-1}(1+q^{u-j}).
\]\end{theorem}
\begin{proof} If $u=i$, the product formula in the theorem is precisely Theorem \ref{rodform}. We now assume that
 $u\geq i+1$. We have
\[\begin{split}\binom{u}{i}_{q^2}\prod_{j=0}^{i-1}(q^i+q^j)
=&\binom{u}{i}_{q}\prod_{j=0}^{i-1}(q^u+q^j)\\
=&\binom{u}{i}_{q}\frac{q^u+1}{q^u+q^{i}}\prod_{j=1}^{i}(q^u+q^j)\\
=&\binom{u}{i}_{q}\frac{q^u+1}{q^{u-i}+1}\prod_{j=0}^{i-1}(q^{u-1}+q^j)\\
=&\binom{u}{i}_{q}\binom{u-1}{i}_{q^2}\binom{u-1}{i}_{q}^{-1}
\frac{q^u+1}{q^{u-i}+1}\prod_{j=0}^{i-1}(q^{i}+q^j).\end{split}
\]
It follows that
\[\binom{u}{i}_{q^2}=\binom{u}{i}_{q}\binom{u-1}{i}_{q^2}\binom{u-1}{i}_{q}^{-1}
\frac{q^u+1}{q^{u-i}+1}.
\]
Hence, by induction,
\[\begin{split}&\binom{u}{i}_{q^2}\sum_{j=0}^iq^{j}\binom{i}{j}_{q^2}\\
=&\binom{u}{i}_{q}\frac{q^u+1}{q^{u-i}+1}\binom{u-1}{i}_{q}^{-1}\binom{u-1}{i}_{q^2}\sum_{j=0}^iq^{j}\binom{i}{j}_{q^2}\\
=&\binom{u}{i}_{q}\frac{q^u+1}{q^{u-i}+1}\prod_{j=0}^{i-1}(1+q^{u-1-j})\\
=&\binom{u}{i}_{q}\prod_{j=0}^{i-1}(1+q^{u-j}).\end{split}
\]
The theorem is proved.
\end{proof}
\section{\small{AN ELIMINATION METHOD}}
In this section we shall establish an elimination method  for the system
\begin{equation}\label{bilinearminimalnumber3}
\sum_{i=1}^u(x_{2i-1}x_{2i}^{2^{(\frac{n}{e}-j)d}}+x_{2i-1}^{2^{(\frac{n}{e}-j)d}}x_{2i})=0,~j=0,1,2,\cdots,s.
\end{equation}
Let $V_{s,u}$ be the set of solutions
$(x_1,x_2,\cdots,x_{2u})\in{\mathbb F}_{2^m}^{2u}$ of the above system. We now prove the following.
\begin{theorem}\label{elimination}The set $V_{s,u}$ is identical to the set of solutions
$(x_1,x_2,\cdots,x_{2u})\in{\mathbb F}_{2^m}^{2u}$ of  the system
\begin{equation}\label{bilinearminimalnumber3e}
\begin{cases}\sum_{i=1}^u(x_{2i-1}x_{2i}^{2^{\frac{nd}{e}}}+x_{2i-1}^{2^{\frac{nd}{e}}}x_{2i})=0,\\
(\tilde{x}_{1}, \tilde{x}_2, \cdots, \tilde{x}_{2u})\in V_{s-1,u},\end{cases}\end{equation}
where $\tilde{x}_i=x_i+x_i^{2^{-d}}$.
\end{theorem}
\begin{proof}
In the system (\ref{bilinearminimalnumber3}), adding $2^{-d}$-th power of the $(j-1)$-th equation to the $j$-th equation, and adding $2^{\frac{nd}{e}}$-th power of the second equation to the first,
we arrive at
\[
\begin{cases}\sum_{i=1}^u(x_{2i-1}x_{2i}^{2^{\frac{nd}{e}}}+x_{2i-1}^{2^{\frac{nd}{e}}}x_{2i})=0,\\
\sum_{i=1}^u(x_{2i-1}x_{2i}^{2^{(\frac{n}{e}-j)d}}+x_{2i-1}^{2^{(\frac{n}{e}-j)d}}x_{2i}
+x_{2i-1}^{2^{-d}}x_{2i}^{2^{(\frac{n}{e}-j)d}}+x_{2i-1}^{2^{(\frac{n}{e}-j)d}}x_{2i}^{2^{-d}})=0,\\
j=0,1,2,\cdots,s.\end{cases}\]
Adding the $(j-1)$-th equation to the $j$-th equation in the above system, we arrive at the system (\ref{bilinearminimalnumber3e}).
The theorem is proved.
\end{proof}
We now apply the above elimination method to prove the following.
\begin{theorem}\label{lineardependence2}If $s\geq u$ and
$(x_1,x_2,\cdots,x_{2u})\in V_{s,u}$, then $x_1,x_2,x_4,x_6,\cdots,x_{2u}$ are linearly dependent over ${\mathbb F}_{2^e}$.
\end{theorem}
\begin{proof}
The lemma is trivial if $u=1$. Now  assume that $u\geq 2$. We may assume that $x_{2u}\neq0$. Then we may further assume that $x_{2u}=1$. By Lemma \ref{elimination}, $(\tilde{x}_{1}, \tilde{x}_2, \cdots, \tilde{x}_{2u})\in V_{s-1,u}$.
As $\tilde{x}_{2u}=0$, we see that $(\tilde{x}_{1}, \tilde{x}_2, \cdots, \tilde{x}_{2u-2})\in V_{s-1,u-1}$. By induction, $(\tilde{x}_{1}, \tilde{x}_2, \tilde{x}_{4}, \tilde{x}_{6}, \cdots, \tilde{x}_{2u-2})$ are linearly dependent over ${\mathbb F}_{2^e}$.
That is, there exists a nonzero vector $(\alpha_{0},\alpha_{1}, \cdots, \alpha_{u-1})\in \mathbb{F}_{2^e}^{u}$ such that $\alpha_{0}\widetilde{x}_{1}+\sum_{i=1}^{u-1}\alpha_{i}\widetilde{x}_{2i}=0$.
So, $$\alpha_{0}x_{1}+\sum_{i=1}^{u-1}\alpha_{i}{x}_{2i}=(\alpha_{0}x_{1}+\sum_{i=1}^{u-1}\alpha_{i}{x}_{2i})^{2^{-d}}.$$ Set $\alpha_u=\alpha_{0}x_{1}+\sum_{i=1}^{u-1}\alpha_{i}{x_{2i}}.$ Then $\alpha_u\in{\mathbb F}_{2^e}$, and $\alpha_{0}x_{1}+\sum_{i=1}^{u}\alpha_{i}{x_{2i}}=0$. Therefore $x_{1}, x_{2}, x_{4}, \cdots, x_{2u}$ are linearly dependent over $\mathbb{F}_{2^e}$. The theorem is proved.
\end{proof}

\section{\small{COUNTING THE NUMBER OF SOLUTIONS}}
\paragraph{}
In this section we shall count the set $V_{s,u}$.
For each $\vec{x}=(x_{1}, x_{2}, \cdots, x_{2u})\in V_{s,u}$, define
\[Z(\vec{x})=\{(c_1,c_2,\cdots,c_u)\in\mathbb{F}_{2^e}^{u}\mid~\sum_{i=1}^uc_ix_{2i}=0\}.\]
For each $\mathbb{F}_{2^e}$-subspace $H$ of $\mathbb{F}_{2^e}^{u}$, define
\[V_{s,u,H}=\{(x_{1}, x_{2}, \cdots, x_{2u})\in V_{s,u}\mid~Z({\vec x})=H\},\]
and
\[W_{s,u,H}=\{(x_{1}, x_{2}, \cdots, x_{2u})\in V_{s,u}\mid~Z({\vec x})\supseteq H\}=\cup_{L\supseteq H}V_{s,u,L}.\]
We  can prove the following.
\begin{lemma}If $H$ is a $\mathbb{F}_{2^e}$-subspace of $\mathbb{F}_{2^e}^{u}$ of dimension $i$, then
\[|W_{s,u,H}|=2^{mi}|V_{s,u-i}|.\]
\end{lemma}
\begin{proof}
Suppose that $H$ is generated by the row vectors of a matrix $A$ over $\mathbb{F}_{2^e}$. We may assume that $A\neq0$. Changing the order of the variables if necessary, we may further the last column of $A$ is $(1,0,\cdots, 0)^{\rm T}$, where ${\rm T}$ denotes the transposition. That is , $A$ is of the form
\[\left( \begin{array}{cc}
\alpha & 1\\
 B&0\\
\end{array} \right),
\]
where $\alpha=(\alpha_{1},\alpha_{2}, \cdots, \alpha_{u-1}) \in \mathbb{F}_{2^{e}}^{u-1}$. Then $W_{s,u,H}$ is the set of solutions $(x_{1}, \cdots, x_{2u})$ in $\mathbb{F}_{2^{m}}^{2u}$ of the system
\[
\begin{cases}
\sum_{i=1}^{u}(x_{2i-1}x_{2i}^{2^{(\frac{n}{e}-j)d}}+x_{2i-1}^{2^{(\frac{n}{e}-j)d}}x_{2i})=0,~j=1, 2, \cdots,s,\\
B(x_{2}, x_{4}, \cdots, x_{2u-2})^{\rm T}=0\\
x_{2u}=\alpha_1x_{2}+\alpha_2x_{4}+\cdots+\alpha_{u-1}x_{2u-2}.
\end{cases}\]
Replacing $x_{2i-1}+\alpha_{i}x_{2u-1}$ with $x_{2i-1}$, $i=1,2,\cdots,u-1$, we conclude that the above system is equivalent to the system
\[
\begin{cases}
\sum_{i=1}^{u-1}(x_{2i-1}x_{2i}^{2^{(\frac{n}{e}-j)d}}+x_{2i-1}^{2^{(\frac{n}{e}-j)d}}x_{2i})=0,~j=0,1, 2, \cdots,s,\\
B(x_{2}, x_{4}, \cdots, x_{2u-2})^{\rm T}=0,\\
x_{2u}=\alpha_1x_{2}+\alpha_2x_{4}+\cdots+\alpha_{u-1}x_{2u-2},\\
x_{2u-1}~{\rm free}.
\end{cases}
\]
The lemma now follows by induction.
\end{proof}
We can also prove the following.
\begin{lemma}\label{independent}For $s\geq u\geq2$, we have
\[|V_{s,u,\{0\}}|=2^{eu(u+1)/2}\prod_{i=0}^{u-1}(2^m-2^{ei}).\]
\end{lemma}
\begin{proof} It suffices to show that
\[|V_{s,u,\{0\}}|=2^{eu}(2^m-2^{e(u-1)})|V_{s,u-1,\{0\}}|.\]
By Theorem \ref{lineardependence2}, it suffices to show that, for each $(\alpha_1,\alpha_2,\cdots,\alpha_u)\in{\mathbb F}_{2^e}^u$, the number of solutions $(x_1,x_{2},\cdots,x_{2u})$ in ${\mathbb F}_{2^m}$ of
the system
\[\begin{cases}
\sum_{i=1}^{u}(x_{2i-1}x_{2i}^{2^{(\frac{n}{e}-j)d}}+x_{2i-1}^{2^{(\frac{n}{e}-j)d}}x_{2i})=0,~j=0,1,\cdots,s,\\
x_{1}=\alpha_1x_{2}+\alpha_2x_{4}+\cdots+\alpha_u x_{2u},\\
x_{2},x_4,\cdots,x_{2u}\text{ are linearly independent over }{\mathbb F}_{2^e}\end{cases}\]
is equal to $(2^m-2^{e(u-1)})V_{s,u-1,\{0\}}$.  Replacing $x_{2i-1}$ with $x_{2i-1}+\alpha_ix_{2}$ for each $i\geq2$, we see that the above system is equivalent to the system
\[\begin{cases}
\sum_{i=2}^{u}(x_{2i-1}x_{2i}^{2^{(\frac{n}{e}-j)d}}+x_{2i-1}^{2^{(\frac{n}{e}-j)d}}x_{2i})=0,~j=0,1,\cdots,s,\\
x_{1}=\alpha_1x_{2}+\alpha_2x_{4}+\cdots+\alpha_u x_{2u},\\
x_{2},x_4,\cdots,x_{2u}\text{ are linearly independent over }{\mathbb F}_{2^e},\end{cases}\]
whose number of solutions is precisely $(2^m-2^{e(u-1)})V_{s,u-1,\{0\}}$.
The lemma now follows.\end{proof}
Applying the $q$-binomial M\"{o}bius inversion formula, we arrive at the following.
\begin{corollary}\label{recursiverelation} If $s\geq u\geq1$, then
\[\sum_{i=0}^{u}(-1)^{u-i}2^{e\binom{u-i}{2}}\binom{u}{i}_{2^e}2^{-mi}|V_{s,i}|
=2^{eu(u+1)/2}\prod_{i=0}^{u-1}(1-2^{ei-m}),
\]where $|V_{s,0}|=1$.
\end{corollary}
Applying the $q$-binomial M\"{o}bius inversion formula once more, we arrive at the following.
\begin{theorem}\label{numberofsolutions}If $s\geq i\geq1$, then
\[|V_{s,i}|
=2^{mi}\sum_{u=0}^{i}\binom{i}{u}_{2^e}2^{eu(u+1)/2}\prod_{j=0}^{u-1}(1-2^{ej-m}).
\]\end{theorem}
\section{\small{A RECURSIVE RELATION}}
In this section we shall prove an useful recursive relation for $|V_{s,i}|$.
We begin with the following.
\begin{lemma}\label{numberofsolutions2}If $s\geq i\geq1$, then
\[2^{-mi}|V_{s,i}|
=\sum_{j=0}^{i}(-1)^{j}4^{e\binom{j}{2}}2^{-mj}\binom{i}{j}_{4^e}\sum_{u=0}^{i-j}2^{eu}\binom{i-j}{u}_{4^e}.
\]\end{lemma}\begin{proof}
By Theorem \ref{numberofsolutions}, and by the $q$-binomial formula, we have

\[2^{-mi}|V_{s,i}|
=\sum_{u=0}^{i}\binom{i}{u}_{2^e}2^{eu(u+1)/2}\sum_{j=0}^{u}(-1)^j2^{e\binom{j}{2}}\binom{u}{j}_{2^e}2^{-mj}.
\]
Changing the order of summation, we arrive at
\[2^{-mi}|V_{s,i}|
=\sum_{j=0}^{i}(-1)^j2^{e\binom{j}{2}}2^{-mj}\sum_{u=j}^{i}2^{eu(u+1)/2}\binom{i}{u}_{2^e}\binom{u}{j}_{2^e}.
\]
Applying the identity
\[\binom{i}{u}_{2^e}\binom{u}{j}_{2^e}=\binom{i}{j}_{2^e}\binom{i-j}{u-j}_{2^e},\]
we arrive at
\[2^{-mi}|V_{s,i}|
=\sum_{j=0}^{i}(-1)^j2^{2e\binom{j}{2}}2^{-mj}\binom{i}{j}_{2^e}
\sum_{u=0}^{i-j}2^{e\binom{u}{2}}\binom{i-j}{u}_{2^e}2^{euj}.
\]
Applying the $q$-binomial formula, we arrive at
\[2^{-mi}|V_{s,i}|
=\sum_{j=0}^{i}(-1)^j2^{ej^2}2^{-mj}\binom{i}{j}_{2^e}
\prod_{u=0}^{i-j-1}(1+2^{e(i-u)}).
\]
Applying Theorem \ref{twovsone}, we arrive at
\[2^{-mi}|V_{s,i}|
=\sum_{j=0}^{i}(-1)^{j}2^{ej^2}2^{-mj}\binom{i}{j}_{4^e}\sum_{u=0}^{i-j}2^{eu}\binom{i-j}{u}_{4^e}.
\]
The lemma is proved.
\end{proof}
We now prove the following.
\begin{theorem}\label{numberofsolutions3}If $s\geq i\geq1$, then
\[\sum_{i=0}^v(-1)^{v-i}4^{e\binom{v-i}{2}}\binom{v}{i}_{4^e}
2^{-mi}|V_{s,i}|=2^{(e-m)v}\prod_{j=0}^{v-1}(2^{m}-4^{ej}).
\]
\end{theorem}\begin{proof}
By Lemma \ref{numberofsolutions2}, we have
\[\begin{split}&\sum_{i=0}^v(-1)^{v-i}4^{e\binom{v-i}{2}}\binom{v}{i}_{4^e}
2^{-mi}|V_{s,i}|\\
=&\sum_{i=0}^v(-1)^{v-i}4^{e\binom{v-i}{2}}\binom{v}{i}_{4^e}
\sum_{j=0}^{i}(-1)^{j}2^{ej^2}2^{-mj}\binom{i}{j}_{4^e}\sum_{u=0}^{i-j}2^{eu}\binom{i-j}{u}_{4^e}
\\
=&\sum_{j=0}^{v}(-1)^{j}2^{ej^2}2^{-mj}\sum_{i=j}^v(-1)^{v-i}4^{e\binom{v-i}{2}}\binom{v}{i}_{4^e}
\binom{i}{j}_{4^e}\sum_{u=0}^{i-j}2^{eu}\binom{i-j}{u}_{4^e}\\
=&\sum_{j=0}^{v}(-1)^{j}2^{ej^2}2^{-mj}\sum_{u=0}^{v-j}2^{eu}\sum_{i=u+j}^v(-1)^{v-i}4^{e\binom{v-i}{2}}\binom{v}{i}_{4^e}
\binom{i}{j}_{4^e}\binom{i-j}{u}_{4^e}.
\end{split}
\]
Applying the identity
\[\binom{v}{u}_{4^e}\binom{v-u}{i-u}_{4^e}\binom{i-u}{j}_{4^e}
=\binom{v}{i}_{4^e}\binom{i}{j}_{4^e}\binom{i-j}{u}_{4^e},\]we arrive at

\[\begin{split}&\sum_{i=0}^v(-1)^{v-i}4^{e\binom{v-i}{2}}\binom{v}{i}_{4^e}
2^{-mi}|V_{s,i}|\\
=&\sum_{j=0}^{v}(-1)^{j}2^{ej^2}2^{-mj}\sum_{u=0}^{v-j}2^{eu}\binom{v}{u}_{4^e}
\sum_{i=j}^{v-u}(-1)^{v-i-u}4^{e\binom{v-i-u}{2}}\binom{v-u}{i}_{4^e}\binom{i}{j}_{4^e}.
\end{split}
\]
Applying the
$q$-binomial M\"{o}bius inversion formula, we arrive at
\[\sum_{i=0}^v(-1)^{v-i}4^{e\binom{v-i}{2}}\binom{v}{i}_{4^e}
2^{-mi}|V_{s,i}|=2^{ev}\sum_{j=0}^{v}(-1)^{j}4^{e\binom{j}{2}}\binom{v}{j}_{4^e}2^{-mj}.
\]
Applying the $q$-binomial formula, we arrive at
\[\sum_{i=0}^v(-1)^{v-i}4^{e\binom{v-i}{2}}\binom{v}{i}_{4^e}
2^{-mi}|V_{s,i}|=2^{(e-m)v}\prod_{j=0}^{v-1}(2^{m}-4^{ej}).
\]\end{proof}
\section{\small{ENTERING BILINEAR EQUATIONS}}
In this section we shall prove Theorem \ref{main}.
We begin with the following.
\begin{theorem}\label{signeffect}For each $r=0,2,\cdots,\frac{m}{e}$, \[\alpha_{r,\varepsilon} =\frac12(2^{er}+\varepsilon2^{\frac{er}{2}})\beta_r,\]
where
\begin{equation}\label{rkfrequencydef}\beta_{r}=2^{-m}\#\{\vec{a}\in{\mathbb F}_{2^n}\times{\mathbb F}_{2^m}^k\mid~{\rm rk}(B_{\vec a})=r,~(a_0,a_1,\cdots,a_{k-1})\neq0\}.\end{equation}
\end{theorem}
\begin{proof} By (\ref{dcfrequencydef}), (\ref{dcexpsumrelation}), (\ref{quadraticexpsum}), (\ref{rkbilinear}), and (\ref{rkfrequencydef}), \[\begin{split}&2^{m-\frac{er}{2}}(\alpha_{r,1} -\alpha_{r,-1})\\&=
\sum_{{\rm rk}(B_{\vec a} )=r}\sum_{x\in{\mathbb F}_{2^m}}(-1)^{{\rm Tr}_{{\mathbb F}_{2^e}/{\mathbb F}_2}(Q_{\vec a} (x))}\\&=2^{-m}\sum_{c\in{\mathbb F}_{2^m}}\sum_{{\rm rk}(B_{\vec a} )=r}\sum_{x\in{\mathbb F}_{2^m}}(-1)^{{\rm Tr}_{{\mathbb F}_{2^e}/{\mathbb F}_2}({\rm Tr}_{{\mathbb F}_{2^m}/{\mathbb F}_{2^e}}(cx)+Q_{\vec a} (x))}\\&=2^{-m}\sum_{{\rm rk}(B_{\vec a} )=r}\sum_{x\in{\mathbb F}_{2^m}}(-1)^{{\rm Tr}_{{\mathbb F}_{2^e}/{\mathbb F}_2}(Q_{\vec a} (x))}\sum_{c\in{\mathbb F}_{2^m}}(-1)^{{\rm Tr}_{{\mathbb F}_{2^m}/{\mathbb F}_{2}}(cx)}\\
&=2^{m}
\beta_r.\end{split}
\]
Similarly,
 \[\begin{split}&2^{2m-er}(\alpha_{r,1} +\alpha_{r,-1})\\&=
\sum_{{\rm rk}(B_{\vec a} )=r}(\sum_{x\in{\mathbb F}_{2^m}}(-1)^{{\rm Tr}_{{\mathbb F}_{2^e}/{\mathbb F}_2}(Q_{\vec a} (x))})^2\\&=2^{-m}\sum_{c\in{\mathbb F}_{2^m}}\sum_{{\rm rk}(B_{\vec a} )=r}(\sum_{x\in{\mathbb F}_{2^m}}(-1)^{{\rm Tr}_{{\mathbb F}_{2^e}/{\mathbb F}_2}({\rm Tr}_{{\mathbb F}_{2^m}/{\mathbb F}_{2^e}}(cx)+Q_{\vec a} (x))})^2\\&=2^{-m}\sum_{{\rm rk}(B_{\vec a} )=r}\sum_{x,y\in{\mathbb F}_{2^m}}(-1)^{{\rm Tr}_{{\mathbb F}_{2^e}/{\mathbb F}_2}(Q_{\vec a} (x)+Q_{\vec a} (y))}\sum_{c\in{\mathbb F}_{2^m}}(-1)^{{\rm Tr}_{{\mathbb F}_{2^m}/{\mathbb F}_{2}}(c(x+y))}\\
&=2^{2m}
\beta_r.\end{split}
\]
The theorem is proved.\end{proof}
By Theorem \ref{signeffect}, Theorem \ref{main} follows from the following.
\begin{theorem}
For each $j=0,1,\cdots,k-1$, we have
\[\beta_{\frac{m}{e}-2j,\varepsilon}=
\sum_{v=j}^{k-1}(-1)^{v-j}4^{e\binom{v-j}{2}}\binom{v}{j}_{4^e}
\binom{\frac{m}{2e}}{v}_{4^e}(2^{n(2k-1-2v)+ev}-1).\]\end{theorem}
\begin{proof}
By the orthogonality of characters, we have
\[
\sum_{{\vec a}\in\mathbb{F}_{2^{n}}\times\mathbb{F}_{2^{m}}^{k}}\big(\sum_{x,y\in{\mathbb F}_{2^m}}(-1)^{{\rm Tr}_{{\mathbb F}_{2^e}/{\mathbb F}_2}(B_{\vec a}(x,y))}\big)^u=2^{n(2k+1)}|V_{k-1,u}|,~0\leq u\leq k-1.
\]
Applying the identity
\[\sum_{x,y\in{\mathbb F}_{2^m}}(-1)^{{\rm Tr}_{{\mathbb F}_{2^e}/{\mathbb F}_2}(B_{\vec a}(x,y))}=2^{2m-e\cdot{\rm rk}(B_{\vec a})},
\]
we arrive at \[\sum_{2\mid r=0}^{m/e}\beta_r2^{u(2m-er)}=2^{n(2k-1)}|V_{k-1,u}|-2^{2mu},~0\leq u\leq k-1.\]
Applying Theorem \ref{rankbound},
we arrive at \[\sum_{2\mid r=m/e-2(k-1)}^{m/e}\beta_r2^{u(2m-er)}=2^{n(2k-1)}|V_{k-1,u}|-2^{2mu},~0\leq u\leq k-1.\]
That is,
\[\sum_{i=0}^{k-1}\beta_{\frac{m}{e}-2i}4^{eiu}=2^{n(2k-1-2u)}|V_{k-1,u}|-2^{mu},~0\leq u\leq k-1.\]
Applying the inversion formula of a symmetric van der Monte matrix, we arrive at\[\beta_{\frac{m}{e}-2j}
=\sum_{v=j}^{k-1}(-1)^{v-j}4^{e\binom{v-j}{2}}\binom{v}{j}_{4^e}\prod_{i=0}^{v-1}(4^{ev}-4^{ei})^{-1}
\sum_{i=0}^v(-1)^{v-i}4^{e\binom{v-i}{2}}\binom{v}{i}_{4^e}y_i,\]
where
\[y_i=2^{n(2k-1-2i)}|V_{k-1,i}|-2^{im}.\]
The contribution of $-2^{im}$ to $\beta_{\frac{m}{e}-2j}$ is
\[\begin{split}
&-\sum_{v=j}^{k-1}(-1)^{v-j}4^{e\binom{v-j}{2}}\binom{v}{j}_{4^e}\prod_{i=0}^{v-1}(4^{ev}-4^{ei})^{-1}
\sum_{i=0}^v(-1)^{v-i}4^{e\binom{v-i}{2}}\binom{v}{i}_{4^e}2^{im}\\
=&-\sum_{v=j}^{k-1}(-1)^{v-j}4^{e\binom{v-j}{2}}\binom{v}{j}_{4^e}
\prod_{i=0}^{v-1}(2^{m}-4^{ei})(4^{ev}-4^{ei})^{-1}\\
=&-\sum_{v=j}^{k-1}(-1)^{v-j}4^{e\binom{v-j}{2}}\binom{v}{j}_{4^e}
\binom{\frac{m}{2e}}{v}_{4^e}.\end{split}\]
By Theorem \ref{numberofsolutions3}, the contribution of $2^{n(2k-1-2i)}|V_{k-1,i}|$ to $\beta_{\frac{m}{e}-2j}$ is equal to
\[2^{n(2k-1)}\sum_{v=j}^{k-1}(-1)^{v-j}4^{e\binom{v-j}{2}}\binom{v}{j}_{4^e}2^{(e-m)v}.\]
\end{proof}
We now prove Theorem \ref{balanced}.
We have
\[\begin{split}&\#\{c\in C\mid ~{\rm DC}(c)=-1\}\\
=&2^{n(2k+1)}-1-
\sum_{j=0}^{k-1}2^{m-2ej}
\sum_{v=j}^{k-1}(-1)^{v-j}4^{e\binom{v-j}{2}}\binom{v}{j}_{4^e}
\binom{\frac{m}{2e}}{v}_{4^e}(2^{n(2k-1-2v)+ev}-1)\\
=&2^{n(2k+1)}-1-
\sum_{v=0}^{k-1}\binom{\frac{m}{2e}}{v}_{4^e}(2^{n(2k-1-2v)+ev}-1)\sum_{j=0}^{v}2^{m-2ej}
(-1)^{v-j}4^{e\binom{v-j}{2}}\binom{v}{j}_{4^e}
\\
=&2^{n(2k+1)}-1-
\sum_{v=0}^{k-1}(-1)^v\binom{\frac{m}{2e}}{v}_{4^e}(2^{n(2k-1-2v)+ev}-1)2^{m-2ev}\prod_{j=1}^{v}
(4^{ej}-1)\\
=&2^{n(2k+1)}-1-
\sum_{v=0}^{k-1}(-1)^v(2^{n(2k-1-2v)+ev}-1)2^{m-ev(v+1)}\prod_{j=0}^{v-1}(2^m-4^{ej})\\
\approx &2^{n(2k+1)}\sum_{v=1}^{k-1}(-1)^{v-1}2^{-ev^2}.\end{split}\]
Theorem \ref{balanced} is proved.
\paragraph{}
{\bf Acknowledgement.} The author thanks Kai-Uwe Schmidt for telling him some background of this subject.

\end{document}